\documentclass{svjour3}                     
\smartqed  
\usepackage{graphicx}
\usepackage{mathptmx}      
\usepackage{latexsym}
\usepackage{graphicx,amsfonts,amssymb,amsmath,graphics,epsfig,pstricks,fancyhdr,fancybox}
\usepackage[T1]{fontenc}

 \newtheorem{stdef}{Definition}
 \newtheorem{prop}{Proposition}

\begin{document}

\title{\LARGE \textbf{Multidimensional Quasi-Monte Carlo Malliavin
Greeks}}

\author{Nicola \textsc{Cufaro Petroni}\\
Dipartimento di Matematica and \textsl{TIRES}, Universit\`a di Bari\\
\textsl{INFN} Sezione di Bari\\
via E. Orabona 4, 70125 Bari, Italy\\
\textit{email: cufaro@ba.infn.it}\\
\\
Piergiacomo \textsc{Sabino}\\
Dipartimento di Matematica, Universit\`a di Bari\\
via E. Orabona 4, 70125 Bari, Italy \\
\textit{email: sabino@dm.uniba.it}}

\maketitle
\def\indic{1\!\!1}
\begin{abstract}
We investigate the use of Malliavin calculus in order to calculate
the Greeks of multidimensional complex path-dependent options by
simulation. For this purpose, we extend the formulas employed by
Montero and Kohatsu-Higa to the multidimensional case. The
multidimensional setting shows the convenience of the Malliavin
Calculus approach over different techniques that have been
previously proposed. Indeed, these techniques may be computationally
expensive and do not provide flexibility for variance reduction. In
contrast, the Malliavin approach exhibits a higher flexibility by
providing a class of functions that return the same expected value
(the Greek) with different accuracies. This versatility for variance
reduction is not possible without the use of the generalized
integral by part formula of Malliavin Calculus. In the
multidimensional context, we find convenient formulas that permit to
improve the localization technique, introduced in Fournié et al
 and reduce both the computational cost and the
variance. Moreover, we show that the parameters employed for
variance reduction can be obtained \textit{on the flight} in the
simulation. We illustrate the efficiency of the proposed procedures,
coupled with the enhanced version of Quasi-Monte Carlo simulations
as discussed in Sabino, for the numerical estimation of the Deltas
of call, digital Asian-style and Exotic basket options with a fixed
and a floating strike price in a multidimensional Black-Scholes
market.
\end{abstract}


\section{Introduction and Motivation}\label{sec:introduction}
Risk-sensitivities, also called Greeks, are fundamental quantities
for the risk-management. Greeks measure the sensitivities of a
portfolio of financial instruments with respect to the parameters of
the underlying model. Mathematically speaking, a greek is the
derivative of a financial quantity with respect to (w.r.t.) any of
the parameters of the problem. As these quantities measure risk, it
is important to calculate them quickly and with a small order of
error. In general, the computational effort required for an accurate
calculation of sensitivities is often substantially greater than
that required for price estimation.

The problem of greeks calculation can be casted as follows. Suppose
that the financial quantity of interest is described by
$\mathbb{E}\left[\psi\left(X(\alpha)\right)Y\right]$ (i.e., the
price of a derivative contract), where
$\psi:\mathbb{R}\rightarrow\mathbb{R}$ is a measurable function and
$X$ and $Y$ are two random variables (r.v.s). The greek, that we
denote $\theta$, is the derivative w.r.t. the parameter $\alpha$:
\begin{equation*}
\theta(\alpha)=\frac{\partial}{\partial
\alpha}\mathbb{E}\left[\psi\left(X(\alpha)\right)Y\right] =
\mathbb{E}\left[\frac{\partial}{\partial
\alpha}\psi\left(X(\alpha)\right)Y\right].
\end{equation*}
The most common of the Greeks are notably, Delta, Gamma, Vega,
Theta, Rho. These quantities are relatively simple to calculate for
plain vanilla contracts in the Black-Scholes (BS) market. However,
their evaluation is a complex and demanding task for exotic
derivative contracts such as Asian-style basket options where no
closed-formula is known.

The simplest and crudest approach is to employ the Monte Carlo (MC)
estimation of $\mathbb{E}\left[\psi\left(X(\alpha)\right)Y\right]$
for two or more values of $\alpha$ and then use finite-difference
approximations. However, this approach can be computationally
intensive and can produce large biases and large variances in
particular if $\psi=\indic_{A}$, where $A$ is a measurable set. A
 variant is the \emph{kernel method} (see Montero and
Kohatsu-Higa \cite{MKH2003}) which generalizes finite-difference
methods using ideas taken from the kernel density estimation.

Several alternatives have been proposed without finite-difference
approximation. \emph{Pathwise methods} (see Glasserman
\cite{Glass2004}) treat the parameter of differentiation $\alpha$ as
a parameter of the evolution of the underlying model and
differentiate this evolution. However, this approach is not always
applicable, notably when $\psi$ is not smooth (for instance
$\psi=\indic_{A}$). At the other extreme, the \emph{likelihood
method ratio} (see Glasserman \cite{Glass2004}) puts the parameter
in the measure describing the underlying model and differentiates
this measure. Even if the likelihood method ratio is applicable to
non-smooth functions it may provide high-variance estimators.
Indeed, compared to the pathwise method (when applicable), it
displays a higher variance. Summarizing, these two alternatives
involve two main ideas: differentiating the evolution or
differentiating the measure, respectively.

In this paper we investigate the use of Malliavin Calculus in order
to employ (Quasi)-Monte Carlo (QMC) simulations for the evaluation
of the sensitivities of complex multidimensional path-dependent
options. The multidimensional setting shows the very convenience of
the Malliavin Calculus approach over the different techniques that
have been proposed. Indeed, Malliavin Calculus allows to calculate
sensitivities as expected values whose estimation is a natural
application of MC methods. Formally:
\begin{equation*}
\theta(\alpha)= \mathbb{E}\left[\frac{\partial}{\partial
\alpha}\psi\left(X(\alpha)\right)Y\right]=\mathbb{E}\left[\psi\left(X(\alpha)\right)H\right].
\end{equation*}
where $H$ is a r.v. depending on $X$ and $Y$.

In the context of multidimensional options, we extend the formulas
employed by Montero and Kohatsu-Higa \cite{MKH2004} to the
multidimensional case. This approach gives a certain flexibility and
provides a class of functions (different r.v.s $H$) returning the
same expected value (the sensitivity) but with different accuracies.

Indeed, the previously mentioned alternative techniques may be
computationally expensive in the multidimensional case and do not
provide flexibility for variance reduction. This versatility for
variance reduction is not possible without the use of the
generalized integral by part formula of Malliavin Calculus. Advanced
techniques such as the kernel density estimation or more recent
approaches such as the Vibrato Monte Carlo in Gilles \cite{Gil09}
are difficult to employ and computationally demanding in
multi-dimensions. In order to avoid to use the Malliavin technique,
Chen and Glasserman \cite{CG07} have illustrated a procedure that
produces ``Malliavin Greeks'' without Malliavin Calculus. However,
since  this procedure involves both pathwise and likelihood ratio
methods, the estimators of the formulas for the sensitivities in
Chen and Glasserman \cite{CG07} have a high variance.

For these purposes we find convenient representations of $H$ that
permit to enhance the localization technique introduced in Fournié
et al. \cite{FLLL1999} and reduce both the computational cost and
the variance. Moreover, we show that the parameters employed for the
variance reduction can be obtained \textit{on the flight} in the
simulation by adaptive techniques. We illustrate the efficiency of
the proposed procedures, coupled with the enhanced version of QMC
simulations discussed in Sabino \cite{Sabino08b}, for the numerical
estimation of the Deltas of call, digital Asian-style and Exotic
basket options with a fixed and a floating strike price in a
multidimensional BS market.

The paper is organized as follows. Section~\ref{sec:Mall} is a short
introduction on Malliavin Calculus, Section~\ref{sec:Greeks} derives
the formulas employed for the computation of the Deltas of call
Asian basket options with floating and fixed strike, Asian digital
options and exotic options. Section~\ref{sec:SimSett} illustrates
the enhanced QMC approach that we adopt and describes in details how
to get the localization parameters with adaptive (Q)MC techniques;
Section~\ref{sec:NumInv} discusses the numerical experiments of the
study and finally Section~\ref{sec:conclusions} summarizes the most
important results and concludes the paper.


\section{Malliavin Calculus: Basic Results and
Notation}\label{sec:Mall} The aim of this section is to briefly
introduce the basic results from Malliavin Calculus and to fix the
notation we adopt in the rest of the paper. For more information on
this subject, we refer the reader to the book by Nualart
\cite{Nu06}.

Consider the probability space
$\left(\Omega,\mathcal{F},\mathbb{P}\right)$ where we define the
$M$-dimensional Brownian motion
$\mathbf{W}(t)=\left(W_1(t),\dots,W_M(t)\right)$, $t\in[0,T]$ and
given $0=t_0^{(n)},t_1^{(n)},\dots,t_n^{(n)}$, divide the interval
$[0,T]$ into $n$ subintervals $I_k=[t_k^{(n)},t_{k+1}^{(n)})$,
$k=0,\dots,n-1$. The superscripts indicates the fineness of the
subdivision of $[0,T]$. Now denote the vector
$\Delta^{(n)}_k=\left(\Delta^{(n)}_{k,1},\dots,\Delta^{(n)}_{k,M}\right)$
where $\Delta^{(n)}_{k,i}= W_i(t_{k+1}^{(n)})-W_i(t_{k}^{(n)})$.

Now consider a smooth function with polynomial growth
$\phi:\mathbb{R}^{n\times M}\rightarrow\mathbb{R}$,
$\phi\in\mathcal{C}_p^{\infty}$, of the form
$\phi=\phi(\Delta^{(n)}_0,\dots,\Delta^{(n)}_{n-1})$. Finally we
consider the following space:
\begin{equation}
\mathcal{S}_n=\left\{\phi(\Delta^{(n)}_0,\dots,\Delta^{(n)}_{n-1});\phi\in\mathcal{C}_p^{\infty}
\right\}\subset\mathcal{L}^2\left(\Omega\right),
\end{equation}
where $\left(\mathcal{S}_n\right)_{n\ge 1}$ form an increasing
sequence in $\mathcal{L}^2\left(\Omega\right)$.
\begin{stdef}
The union $\mathcal{S}=\left(\bigcup_{n\ge
1}\mathcal{S}_n\right)\subset\mathcal{L}^2\left(\Omega\right)$ is
called simple functional space and its elements are called simple
functionals.
\end{stdef}
We now can define the Malliavin derivative operator.
\begin{stdef}
Let $\phi\in\mathcal{S}$, then there exists $n\in\mathbb{N}^*$ such
that $\phi=\phi(\Delta^{(n)})$. The Malliavin derivative operator
$\mathbf{D}=\left(D^1,\dots,D^M\right)$ of $\phi$ at a point $s\in
I_k$ is defined as
\begin{equation}
D_s^m\phi = \sum_{k=0}^{n-1}\frac{\partial \phi}{\partial
x_{k,m}}(\Delta^{(n)}_0,\dots,\Delta^{(n)}_{n-1})\indic_{I_k^{(n)}}(s).
\end{equation}
Let us precise the notation. We have
$\mathbf{x}_k=\left(\mathbf{x}_{k,1},\dots,\mathbf{x}_{k,M}\right)$,
so $\mathbf{x}_k$ corresponds to the increment vector
$\Delta_{k}^{(n)}$ and the $m$-th component $x_{k,m}$ corresponds to
the $m$-th component $\Delta^{(n)}_{k,m}$.
\end{stdef}
Moreover we give the following definition.
\begin{stdef}
Introduce on $\mathcal{S}$ the norm
\begin{equation*}
\|F\|^2_{1,2}=\|F\|_{\mathcal{L}^2(\Omega)}+\|\mathbf{D}F\|_{\mathcal{L}^2([0,T]\times\Omega)},
\end{equation*}
the set $\mathbb{D}^{1,2}$ is the closure of $\mathcal{S}$ with
respect (w.r.t.) $\|\cdot\|_{1,2}$.
\end{stdef}
Finally we define the Skorohod integral $\delta^{\text{Sk}}$.
\begin{stdef}
The adjoint  of $\mathbf{D}$ in $\mathcal{L}^2\left(\Omega\times
[0,T]\right)$ is the operator:
\begin{equation}
\delta^{\text{Sk}}\,:\,\mathbf{u}=(u_1,\dots,u_m)\in\mathrm{dom}\left(\delta^{\text{Sk}}\right)\rightarrow
\delta^{\text{Sk}}\left(\mathbf{u}\right)
\end{equation}
which by definition satisfies for $\phi\in\mathbb{D}^{1,2}$ and
$\mathbf{u}\in\mathrm{dom}\left(\delta^{\text{Sk}}\right)$
\begin{equation}\label{eq:dual}
\sum_{m=1}^M\mathbb{E}\left[\int_0^TD_s^m\left(\phi
u_m(s)\right)ds\right]=
\mathbb{E}\left[\phi\sum_{m=1}^M\delta^{\text{Sk}}_m\left(u_m\right)\right]=\mathbb{E}\left[\phi\delta^{\text{Sk}}\left(\mathbf{u}\right)\right].
\end{equation}
\end{stdef}
Equation (\ref{eq:dual}) is known as duality relation.

It can be shown (see Nualart \cite{Nu06}) that if $\mathbf{u}(t)$ is
an Ito process, the Skorohod integral coincides with the Ito
integral of $\mathbf{u}$ and $\mathbf{D}_s\mathbf{u}=0$ if $s\ge t$
.

We now list some identities and useful results that will be employed
in the rest of this paper. Proofs can be found in Nualart
\cite{Nu06}.
\begin{enumerate}
\item $\forall F_1,\dots,F_d\in\mathbb{D}^{1,2}$ we have
$\phi(F_1,\dots,F_d)\in\mathbb{D}^{1,2}$ and $\forall m=1,\dots,M$
\begin{equation}
D_s^m\phi=\sum_{k=1}^d\frac{\partial\phi}{\partial
x_k}\left(F_1,\dots,F_d\right)D_s^mF_k.
\end{equation}
For example, let $\mathbf{a}\in\mathbb{R}^M$ we have
\begin{equation}
D_s^m\exp\left(\sum_{i=1}^Ma_iW_i(t)\right)=a_m\exp\left(\sum_{i=1}^Ma_iW_i(t)\right)\indic_{[0,t]}(s)
\end{equation}
\item Let $\phi(t)$ be an adapted process we have:
\begin{equation}
D_s^m\int_0^T\phi(t)dW^m(t)=\phi(s) + \int_s^TD_s^m\phi(t)dW^m(t),
\end{equation}
and
\begin{equation}
D_s^m\int_0^T\phi(t)dt=\int_s^TD_s^m\phi(t)dt,
\end{equation}

\item If $\phi\in\mathbb{D}^{1,2}$,
$\mathbf{u}\in\mathrm{dom}\left(\delta^{\text{Sk}}\right)$, and
$\phi\mathbf{u}\in\mathrm{dom}\left(\delta^{\text{Sk}}\right)$ then
\begin{equation}\label{eq:SkorProp}
\delta^{\text{Sk}}\left(\phi\mathbf{u}\right)=\phi\delta^{\text{Sk}}\left(\mathbf{u}\right)-\sum_{m=1}^M\int_0^Tu_m(s)D_s^m\phi
ds.
\end{equation}
\end{enumerate}


\section{Multidimensional Malliavin Sensitivities}\label{sec:Greeks}
Consider for simplicity a complete market whose risky assets,
$S_i,\,i=1,\dots,M$, are driven by the following dynamics (in the
risk-neutral measure):
\begin{eqnarray}
dS_i(t)&=&rS_i(t)dt + S_i(t)\sigma_i(t)dB_i(t)\quad i=1,\dots,M,\\
S_i(0)&=&x_i,\nonumber
\end{eqnarray}
where $r$ is the constant risk-free rate,
$\sigma(t)=\left(\sigma_1(t),\dots,\sigma_M(t)\right)$ is the vector
of the volatilities process and $\mathbf{B}(t)$ is the vector of the
$M$-dimensional Brownian motion in the risk-neutral measure with
$dB_{i}(t)dB_m(t)=\rho_{im}(t)dt$; $\rho$ is the correlation matrix
among the Brownian motions (it can be stochastic). The existence of
the vector process $\sigma(t)$ is guaranteed by theorem 9.2.1 in
Shreve \cite{Shreve04}. Applying the risk-neutral pricing formula
(see Shreve \cite{Shreve04}), the calculation of the price at  time
$t$ of any European derivative contract with maturity date $T$ boils
down to the evaluation of an (discounted) expectation:
\begin{equation}
a(t) =
\exp\left(-r(T-t)\right)\mathbb{E}\left[\psi\right|\mathcal{F}_t]\label{1.3}\text{,}
\end{equation}
\noindent  the expectation is under the risk-neutral probability
measure and $\psi$ is a generic $\mathcal{F}_T$-measurable variable
that determines the payoff of the contract.

In order to apply Malliavin Calculus we need to write the above
dynamics in terms of uncorrelated Brownian motions:
\begin{eqnarray}\label{eq.Dyn}
dS_i(t)&=&rS_i(t)dt +
S_i(t)\sigma_i(t)\sum_{m=1}^M\alpha_{im}(t)dW_m(t)\quad
i=1,\dots,M\nonumber\\
S_i(0)&=&x_i,\nonumber
\end{eqnarray}
where $\sum_{m=1}^M\alpha_{im}(t)\alpha_{km}(t)=\rho_{ik}(t), a.s.$
and we have defined
$\sigma_{im}(t)=\sigma_i(t)\sum_{m=1}^M\alpha_{im}(t), a.s.$.

Hereafter we denote $\delta^{\mathrm{Kr}}$ and $\delta^{\mathrm{D}}$
the Kronecker delta and the Dirac delta, respectively. Naturally at
time $T$ we have $\psi=a(T),\,a.s.$.

The following proposition generalizes the formula in Montero
Kohatsu-Higa \cite{MKH2004} to the multidimensional case.
\begin{prop}\label{prop:1}
Assuming the dynamics (\ref{eq.Dyn}) let $m(T)$ be a
$\mathcal{F}_T$-measurable r.v. (it can depend on the entire
trajectory) and consider $\psi = \psi(m(T))$. Denote $G_k$ the
partial derivative
\begin{equation}
G_k=\frac{\partial m(T)}{\partial x_k},\quad k=1,\dots,M,
\end{equation}
Suppose that $\psi\in\mathbb{D}^{1,2}$,  the $k$-th delta (the
$k$-th component of the gradient) is
\begin{equation}\label{eq:delta}
\Delta_k =\frac{\partial a(0)}{\partial x_k}=e^{-rT}
\mathbb{E}\left[\psi'G_k\right]=e^{-rT}\mathbb{E}\left[\psi\sum_{m=1}^M\delta_m^{\mathrm{Sk}}(G_ku_m)\right],
\end{equation}
\noindent where $\mathbf{u}=(u_1,\dots,u_M)\in
\mathrm{dom}(\delta^{\mathrm{Sk}})$, $\mathbf{z}=(z_1,\dots,z_m)\in
\mathrm{dom}(\delta^{\mathrm{Sk}})$, $G_k\mathbf{u}\in
\mathrm{dom}(\delta^{\mathrm{Sk}})$ and
\begin{eqnarray*}
 u_m(s) &=& \frac{z_m(s)}{\sum_{h=1}^M\int_0^Tz_h(s)D^h_s m(T)ds}  \\
  \sum_{h=1}^M\int_0^Tz_h(s)D^h_s m(T)ds &\neq& 0,\quad \mathrm{a.s.}
\end{eqnarray*}
\end{prop}
The derivative $\psi'$ may have no mathematical sense indeed, the
aim of the proposition is to overtake the problem with the formalism
of distributions and Malliavin Calculus.
\begin{proof}
Compute
\begin{equation}
D_s^h\psi=\psi'D^h_s m(T)\quad h=1,\dots,M.
\end{equation}
Suppose $\mathbf{z}\in \mathrm{dom}(\delta^{\mathrm{Sk}})$ and
multiply the above equation by $z_h(t)$ and by $G_k$; then sum for
all $h=1,\dots,M$ and integrate:
\begin{equation}
\sum_{h=1}^M\int_0^TG_kz_h(s)D_s^h\psi(T)ds=\sum_{h=1}^M\int_0^TG_kz_h(s)\psi'(T)D^h_s
m(T)ds.
\end{equation}
$\psi'G_k$ does not depend on $s$ and due to the definition of
$\mathbf{u}$ we can write
\begin{equation}\label{eq:proofMall}
\psi'G_k=\sum_{m=1}^M\int_0^Tu_m(s)G_kD_s^m\psi(T))ds.\quad
k=1,\dots,M
\end{equation}
Finally compute the expected value of both sides of
(\ref{eq:proofMall})
\begin{equation}
\mathbb{E}\left[\psi'G_k\right]=\mathbb{E}\left[\sum_{m=1}^M\int_0^Tu_m(s)G_kD_s^m\psi
ds\right].
\end{equation}
By duality
\begin{equation}\label{DeltaMulti}
\Delta_k =
\mathbb{E}\left[\psi\delta^{\mathrm{Sk}}(G_k\mathbf{u})\right]\quad
k=1,\dots,M,
\end{equation}
and this concludes the proof.
\end{proof}


\subsection{Greeks in the Multidimensional Black-Scholes
Market}\label{subsec:BS} In this section we apply Proposition
\ref{prop:1} to the case of a multidimensional Black-Scholes market
where the volatilities vector process in Equation (\ref{eq.Dyn}) is
not stochastic (for simplicity we consider constant volatilities and
correlations). The main advantage of the Malliavin approach over
different techniques, for example the methods in Gilles \cite{Gil09}
and the Chen and Glasserman \cite{CG07}, is that Proposition
\ref{prop:1} allows the possibility of variance  and computational
reduction due to the flexibility in choosing either the process
$\mathbf{u}$, or better $\mathbf{z}$. The methods illustrated in
Gilles \cite{Gil09} and Chen and Glasserman \cite{CG07} are
difficult to employ if we assume a multidimensional dynamics and
they do not allow versatility for variance reduction.

We consider the case $z_h=\alpha_k\delta^{\mathrm{Kr}}_{hk}$;
$h,k=1,\dots,M$, $\alpha_k=1, \forall k$. Namely, in order to
compute the $k$-th delta we consider only the $k$-th term of the
Skorohod integral reducing the computational cost. In particular,
this choice is motivated by the fact that  we can enhance the
localization technique introduced by Fourn\'{e} et al.
\cite{FLLL1999}. With this setting we need to control only
$\delta^{\mathrm{Sk}}_{k}(\cdot)$ and then only the $k$-th component
of $\mathbf{W}(t)$. This enhancement is not possible with other
approaches that furnish only a fixed representation of the
components of the multidimensional deltas.

Under the above assumptions for the vector process $\mathbf{z}$, we
explicitly derive the multidimensional deltas  for the following
exotic options in the BS market:
\begin{enumerate}
\item Discretely monitored Asian basket options with fixed strike. Assume
$t_1<t_2\dots<t_N=T$, where $T$ is the maturity of the contract and
the payoff function
\begin{equation}
\psi =\left(\sum_{i=1}^{M}\sum_{j=1}^{N}w_{ij}\,S_{i}\left(
t_{j}\right) -K\right)^+,
\end{equation}%
\noindent where $K$ is the strike price and $\sum_{i,j}w_{ij}=1$. In
this case we have
\begin{equation*}
G_k = \frac{1}{x_k}\sum_{j=1}^{N}w_{kj}S_k(t_j)
\end{equation*}
and
\begin{equation*}
m(T) = \sum_{i=1}^{M}\sum_{j=1}^{N}w_{ij}\,S_{i}\left( t_{j}\right).
\end{equation*}
 We then calculate the following quantities
\begin{eqnarray*}\label{prova}
L_k\!\!\!&=&\!\!\!\int_0^TD_s^km(T)ds=\sum_{i=1}^M\sum_{j=1}^Nw_{ij}S_i(t_j)t_j\sigma_{ik},\label{prova1}\\
A_k=\int_0^TD^k_sG_kds\!\!\! &=&\!\!\!
\sum_{j=1}^Nw_{jk}S_k(t_j)t_j\sigma_{kk}(s)=\sum_{j=1}^Nw_{jk}S_k(t_j)t_j\sigma_{k}\label{prova2}\\
B_k=\int_0^TD^k_sL_kds\!\!\!
&=&\!\!\!\sum_{j=1}^Nw_{ij}S_i(t_j)t_j^2\sigma_{ik}^2,\label{prova3}
\end{eqnarray*}
and hence
\begin{equation}\label{DelAsMultiDisCorr}
\Delta_k =
\mathbb{E}\left[\psi\delta_k^{\mathrm{Sk}}\left(\frac{G_k}{L_K}\right)\right],\quad
k=1,\dots,M.
\end{equation}
Due to the equation (\ref{eq:SkorProp}) we can write the  the
Skorohod integral above for $\,k=1,\dots,M$ as:
\begin{equation}
\delta_k\left(\frac{G_k}{L_K}\right)=\frac{G_k}{L_K}\,W_k(T)-
\frac{1}{L_k^2}\left(L_k\int_0^TD^k_sG_kds-G_k\int_0^TD^k_sL_kds\right)=\frac{G_k}{L_K}\left(W_k(T)+\frac{B_k}{L_k}\right)-\frac{A_k}{L_k}.
\end{equation}
With another choice of $\mathbf{z}$, for instance $z_h=\alpha_h$,
$\Delta_k$ would depend linearly on the whole $M$-dimensional
Brownian motion, making the localization technique less efficient.
\item Discretely monitored Asian basket options with floating strike $K(T) =
\frac{\sum_{i=1}^MS_{i}(T)}{M}$. For simplicity we assume
$w_{ij}=\frac{1}{MN}\forall i,j$. The calculation is similar to the
previous payoff function, indeed we can write $\psi=n(T)^+$ where
\begin{equation*}
n(T) = m(T) - K(T).
\end{equation*}
In analogy, we have
\begin{eqnarray*}
F_k&=&\frac{\partial n(T)}{x_k}=G_k-\frac{S_k(T)}{Mx_k}=G_k-T_k,\\
M_k&=&\int_0^TD_s^kn(T)ds=L_k-\int_0^TD_s^kK(T)=L_k-\frac{\sum_{i=1}^MS_i(T)T\sigma_{ik}}{M}=L_k-U_k,\\
\int_0^TD^k_sF_kds\!\!\!
&=&\!\!\!A_k-\int_0^TD_s^kT_kds=A_k-\frac{S_k(T)T\sigma_{k}}{Mx_k}=A_k-V_k,\\
\int_0^TD^k_sM_kds\!\!\! &=&\!\!\!\int_0^TD^k_sL_kds -
\int_0^TD^k_sU_kds=B_k-\frac{\sum_{i=1}^MS_i(T)T^2\sigma_{ik}^2}{M}=B_k-P_k,
\end{eqnarray*}
with the quantities $T_k,U_k,V_k,P_k,\forall k$ automatically
defined by the above equations. Then
\begin{equation}
\Delta_k =
\mathbb{E}\left[\psi\delta_k^{\mathrm{Sk}}\left(\frac{F_k}{M_K}\right)\right],\quad
k=1,\dots,M.
\end{equation}
and
\begin{equation}
\delta_k\left(\frac{F_k}{M_K}\right)=\frac{F_k}{M_K}\left(W_k(T)+\frac{B_k-P_k}{M_k}\right)-\frac{A_k-V_k}{M_k}.
\end{equation}
\item Digital Asian basket options with fixed strike.
\begin{equation}
\psi =\indic_{m(T)\ge K}.
\end{equation}%
This type of payoff function fulfills the hypotheses of Proposition
\ref{prop:1} and we might adopt equation (\ref{DelAsMultiDisCorr}).
However, due to the properties of the Dirac delta
$\delta^{\text{D}}$ and Proposition \ref{prop:1} we can write
\begin{equation*}
\Delta_k =e^{-rT}
\mathbb{E}\left[\delta^{\text{D}}_K\left(m(T)\right)G_k\right]=e^{-rT}
\mathbb{E}\left[\delta^{\text{D}}_K\left(m(T)\right)\phi\left(\frac{m(T)-K}{h}\right)G_k\right]=e^{-rT}\mathbb{E}\left[\psi\sum_{m=1}^M\delta_m^{\mathrm{Sk}}(G_k\phi
u_m)\right],
\end{equation*}
where we assume that $\phi$, $\phi'$ are square integrable,
$\phi(0)=1$, $\phi G_k$ is Skorohod integrable $\forall k=1,\dots,M$
and $h>0$. The aim of this setting is to reduce the variance of the
MC estimator of $\Delta_k$ by tuning the localization function
$\phi$ around the strike $K$ with a convenient choice of the
parameter $h$ (see Kohatsu-Higa and Patterson \cite{KP2002}).

Under this assumption the Skorohod integral in equation
(\ref{eq:delta}) becomes:
\begin{equation*}
\delta^{\text{Sk}}\left(\phi\left(\frac{m(T)-K}{h}\right) G_k
\mathbf{u}\right) = \phi\left(\frac{m(T)-K}{h}\right)
G_k\delta^{\text{Sk}}(\mathbf{u})-\sum_{m=1}^M\int_0^Tu_m(s)D_s^m\left(\phi
G_k\right)ds
\end{equation*}
where for $m=1,\dots,M$
\begin{equation*}
D_s^m\left(\phi\left(\frac{m(T)-K}{h}\right) G_k\right) =
\phi\left(\frac{m(T)-K}{h}\right)
D_s^mG_k+\frac{G_k}{h}\phi'\left(\frac{m(T)-K}{h}\right)D_s^mm(T),
\end{equation*}
then the Skorohod integral
$\delta^{\text{Sk}}\left(\phi\left(\frac{m(T)-K}{h}\right) G_k
\mathbf{u}\right)$ is
\begin{equation}
\phi\left(\frac{m(T)-K}{h}\right)
G_k\delta^{\text{Sk}}(\mathbf{u})-\frac{\phi\left(\frac{m(T)-K}{h}\right)\sum_{m=1}^M\int_0^Tu_m(s)D_s^mG_kds}{\sum_{m=1}^M\int_0^Tu_m(s)D_s^mm(T)ds}
-\frac{G_k}{h}\phi'\left(\frac{m(T)-K}{h}\right).
\end{equation}
Finally, with our choice for the simple process $\mathbf{u}$ the
last equation becomes:
\begin{equation}
\phi\left(\frac{m(T)-K}{h}\right)
G_k\delta^{\text{Sk}}(\mathbf{u})-\frac{\phi\left(\frac{m(T)-K}{h}\right)A_k}{L_k}
-\frac{G_k}{h}\phi'\left(\frac{m(T)-K}{h}\right),
\end{equation}
where $\delta^{\text{Sk}}(\mathbf{u})$ depends on the terms that we
have found in the case of Call Asian basket options.

It is worthwile to say that the same localization procedure and the
Malliavin approach adopted for digital options can be employed for
the computation the Gamma (second order derivative) for Call Asian
basket options.
\end{enumerate}

\subsection{Greeks for Exotic Options}\label{subsec:GreekExotic}
In Proposition \ref{prop:1} we have supposed that the payoff
function $\psi$ depends on $m(T)$ only. With the notation adopted in
the BS setting, suppose for instance, that
$\psi=\max\left(m(T)-K,K(T)-K,0\right)$, where $K$ is a fixed price,
now we cannot rely on Proposition \ref{prop:1} to derive the
expression of the sensitivities of such an exotic option. Here
$\psi$ depends separately on two random variables $m(T)$ and $K(T)$.
In the following we extend Proposition \ref{prop:1} in order to
allow such a dependence.
\begin{prop}\label{prop2} Assuming the dynamics (\ref{eq.Dyn}) suppose $\psi=\psi(X,Y)$. For simplicity we set $r=0$, denote
$G_k=\frac{\partial Y}{\partial x_k}$ and $T_k=\frac{\partial
X}{\partial x_k}$. Let $\mathbf{u}$ and $\mathbf{p}$ be two simple
processes belonging to $\mathrm{dom}(\delta^{\mathrm{Sk}})$. Define
the following $\mathcal{F}_T$-measurable r.v.s:
\begin{eqnarray}
  a_1= \sum_{m=1}^M\int_0^Tu_m(s)D_s^mX
ds,& a_2= \sum_{m=1}^M\int_0^Tu_m(s)D_s^mY ds\\
  b_1 =\sum_{m=1}^M\int_0^Tp_m(s)D_s^mX,& b_2= \sum_{m=1}^M\int_0^Tp_m(s)D_s^mY
  ds\\
  O_1=\sum_{m=1}^M\int_0^TT_ku_m(s)D_s^m\psi ds, &   O_2=\sum_{m=1}^M\int_0^TG_kp_m(s)D_s^m\psi
  ds\\
  U_1 = \frac{b_2-\frac{b_1G_k}{T_k}}{a_1b_2-a_2b_1}, & U_2 =
  \frac{\frac{a_2T_k}{G_k}-a_1}{a_1b_2-a_2b_1}.
\end{eqnarray}
Finally, suppose that $a_1b_2-a_2b_1\ne 0, a.s.$ and
$U_1T_k\mathbf{u}-U_2G_k\mathbf{p}$ is Skorohod integrable, we have:
\begin{equation}
\Delta_k=\frac{\partial \mathbb{E}\left[\psi(X,Y)\right]}{\partial
x_k}=\Delta_k=
\mathbb{E}\left[T_k\partial_X\psi(X,Y)+G_k\partial_Y\psi(X,Y)\right]=\mathbb{E}\left[\psi(X,Y)\sum_{m=1}^M\delta^{\text{Sk}}_m\left(
U_1T_ku_m-U_2G_kp_m\right)\right],
\end{equation}
\noindent where $\partial_X$ and $\partial_Y$ denote the partial
derivatives with respect to the first and second variable,
respectively.
\end{prop}
\begin{proof}
Compute:
\begin{equation}
D_s^m\psi(X,Y) = \partial_X\psi D_s^mX+\partial_Y\psi D_s^mY.
\end{equation}
As done in the proof of Proposition \ref{prop:1}, multiply for $T_k$
and $u_m$, sum for all $m$ and integrate:
\begin{equation}\label{eq1:prop2}
\sum_{m=1}^M\int_0^TT_ku_m(s)D_s^m\psi
ds=\sum_{m=1}^M\int_0^TT_ku_m(s)D_s^mX
ds+\sum_{m=1}^M\int_0^TT_ku_m(s)D_s^mY ds.
\end{equation}
Now repeat the procedure above considering $G_k$ and $p_h(s)$, we
have
\begin{equation}\label{eq2:prop2}
\sum_{m=1}^M\int_0^TG_kp_m(s)D_s^m\psi
ds=\sum_{m=1}^M\int_0^TG_kp_m(s)D_s^mX
ds+\sum_{m=1}^M\int_0^TG_kp_m(s)D_s^mY ds.
\end{equation}
We rewrite Equations (\ref{eq1:prop2}) and (\ref{eq1:prop2}) as a
linear system
\begin{equation}\label{sys:prop2}
\left\{\begin{array}{cc}
         O_1 =& a_1 T_k\partial_X\psi + a_2 T_k\partial_Y\psi\\
         O_2 =& b_1 G_k\partial_X\psi + g_2 G_k\partial_Y\psi
       \end{array}
\right.
\end{equation}
Our aim is to compute
$T_k\partial_X\psi(X,Y)+G_k\partial_Y\psi(X,Y)$ such that we can
apply the duality relation. After some algebra we get that
\begin{equation*}
T_k\partial_X\psi =
\frac{b_2G_kO_1-a_2T_kO_2}{G_k\left(a_1b_2-a_2b_1\right)},
\end{equation*}
\begin{equation*}
G_k\partial_Y\psi =
\frac{a_1T_kO_2-b_1G_kO_1}{T_k\left(a_1b_2-a_2b_1\right)}
\end{equation*}
and
\begin{equation}
T_k\partial_X\psi + G_k\partial_Y\psi=
\frac{O_1\left(b_2-\frac{b_1G_k}{T_k}\right)-O_2\left(\frac{a_2T_k}{G_k}-a_1\right)}{a_1b_2-a_2b_1},
\end{equation}
then we have
\begin{equation}
\Delta_k=\mathbb{E}\left[O_1U_1-O_2U_2\right]=\mathbb{E}\left[\sum_{m=1}^M\int_0^T\left(
U_1T_ku_m(s)-U_2G_kp_m(s)\right)D_s^m\psi ds\right],
\end{equation}
by duality
\begin{equation}
\Delta_k=\mathbb{E}\left[\psi\sum_{m=1}^M\delta^{\text{Sk}}_m\left(U_1T_ku_m-U_2G_kp_m\right)\right]
\end{equation}
and this concludes the proof.
\end{proof}
We can adapt the result of Proposition \ref{prop2} to the BS market.
Again the Malliavin Calculus approach is very versatile and permits
to reduce the computational burden and the variance of the MC by
enhancing the localization technique. As done before we consider
$u_m(s)=\delta^{\text{Kr}}_{mk},\forall s$ and
$p_m(s)=s\delta^{\text{Kr}}_{mk},\forall s$, in order to fulfill the
hypothesis of Proposition \ref{prop2}.

The formula for the $k$-th component of the delta is
\begin{equation}
\Delta_k=\mathbb{E}\left[\psi(X,Y)\left(\delta^{\text{Sk}}_k\left(
U_1T_k\right)-\delta^{\text{Sk}}_k\left(sU_2G_k\right)\right)\right],
\end{equation}
and the two Skorohod integrals are respectively:
\begin{equation}\label{sk1}
\delta^{\text{Sk}}_k\left(
U_1T_k\right)=U_1T_kW_k(T)-U_1\int_0^TD_s^kT_kds-T_k\int_0^TD_s^kU_1ds,
\end{equation}
\begin{equation}\label{sk2}
\delta^{\text{Sk}}_k\left(
U_2G_k\right)=U_2G_k\int_0^TsdW^k_s-G_k\int_0^TsD^k_sU_2ds-U_2\int_0^TsD^k_sG_kds.
\end{equation}
In the MC estimation we can simulate the first term in the above
equation relying on the equality:
\begin{equation*}
\int_0^TsdW^k(s) = TW_k(T)-\int_0^TW_k(s)ds,
\end{equation*}
where $\int_0^TW_k(s)ds$ is approximated by a sum at the points
$t_1,\dots,t_N=T$.

In our numerical experiments we consider
$\psi=\max\left(m(T)-K,K(T)-K,0\right)$ where $m(T)$ and $K(T)$ have
been defined in Section \ref{subsec:BS}. The terms in Equations
(\ref{sk1}) and (\ref{sk2}) have been obtained as in Section
\ref{subsec:BS}.

\section{Simulation Setting}\label{sec:SimSett}
In this section we briefly describe the numerical setting that we
adopt for the QMC estimation of the Greeks by the Malliavin approach
formulas. We briefly illustrate the QMC method and discuss how to
conveniently find  the parameters of the localization technique
\emph{on the fly} by adaptive simulation.


\subsection{The Quasi-Monte Carlo Framework}\label{subsec:QMC}
Consider $I=\mathbb{E}[\psi(X)]$ where $\mathbf{X}$ is a
$d$-dimensional random vector and
$\psi:\mathbb{R}^d\rightarrow\mathbb{R}$, the QMC estimator of $I$
is $\hat{I}_{QMC}=\frac{\sum_{n=1}^{N_S}\psi(\mathbf{X}_n)}{N_S}$,
$N_S$ is the number of simulations, as for the standard MC. However
the points $\mathbf{X}_i$ are not pseudo-random  but are obtained by
low-discrepancy sequences. Low-discrepancy sequences do not mimic
randomness but display better regularity and distribution (see
Niederreiter \cite{Ni1992} for more on this subject). We do not
enter into the details of QMC methods and their properties, we just
stress the fact that such techniques do not rely on the central
limit theorem and the error bounds are given by the well known
Hlawka-Koksma inequality. Some randomness is then introduced in
order to statistically estimate the error of the estimation by the
sampled variance; this task is achieved by a technique called
\emph{scrambling} (see Owen \cite{ow2002}). The randomized version
of QMC is called Randomized Quasi-Monte Carlo (RQMC).

In our numerical estimation we use a randomized version of the
Sobol' sequence with Sobol's property A, that is one of the most
used low-discrepancy sequences (it is also a digital net).

Finally, in order to improve the efficiency of RQMC and reduce the
effect of the so-called \emph{curse of dimensionality}, we employ
the Linear Transformation (LT) technique introduced in Imai and Tan
\cite{IT2006} in the enhanced version illustrated in Sabino
\cite{Sabino08b,Sabino09}. The aim of the LT algorithm is to
concentrate the variance of $\psi$ into the components with higher
variability so that we may profit from the higher regularity of
low-discrepancy points and then reduce the nominal dimension of
$\psi$.

We briefly describe the LT algorithm. Consider a $d$ dimensional
normal random vector $\mathbf{T}\sim \mathcal{N}(\mu;\Sigma)$, a
vector $\mathbf{w}=(w_1,\dots,w_d)\in\mathbb{R}^d$ and let
$f(\mathbf{T}) =\sum_{i=1}^d w_iT_i$ be a linear combination of
$\mathbf{T}$. Let $C$ be such that $\Sigma=CC^T$ and assume
$\epsilon\sim\mathcal{N}(0,I_d)$ with
$\mathbf{T}\stackrel{\mathcal{L}}{=}C\epsilon$. The LT approach
considers $C$ as $C=C^{\text{LT}}=C^{\text{CH}}A$, with
$C^{\text{CH}}$ the Cholesky decomposition of $\Sigma$. Then, in the
linear case, we can define:
\begin{equation}\label{4.3.4}
    g^{A}(\epsilon):= f(C^{\text{CH}}A\epsilon) = \sum_{k=1}^d \alpha_k \epsilon_k + \mu\cdot
    \mathbf{w},
\end{equation}
\noindent where $\alpha_k= \mathbf{C^{\text{LT}}_{\cdot k}}\cdot
\mathbf{w}=\mathbf{A_{\cdot k}}\cdot\mathbf{B},\,k=1\dots,d$ and
$\mathbf{B}=(C^{\text{CH}})^T\mathbf{w}$ while $\mathbf{C_{\cdot
k}}$ and $\mathbf{A_{\cdot k}}$ are the $k$-th columns of the matrix
$C$ and $A$, respectively. In the linear case, setting
\begin{equation}\label{4.3.8}
    \mathbf{A_{\cdot 1}^*} = \pm\frac{\mathbf{B}} {\|\mathbf{B}\|},
\end{equation}
with arbitrary remaining columns with the only constrain that
$AA^T=I_d$, leads to the following expression:
\begin{equation}
g^A(\epsilon)=\mu\cdot\mathbf{w}\pm\|\mathbf{B}\|\epsilon_1.
\end{equation}
This is equivalent to reduce the effective dimension in the
truncation sense to $1$ and this means to maximize the variance of
the first component $\epsilon_1$.

In a non-linear framework, we can use the LT construction, which
relies on the first order Taylor expansion of $g^A$:
\begin{equation}\label{4.3.12}
    g^A(\epsilon) \approx g^A(\hat{\epsilon}) +
    \sum_{l=1}^d\frac{\partial
    g^A(\hat{\epsilon})}{\partial\epsilon_l}\Delta\epsilon_l.
\end{equation}
\noindent The approximated function is linear in the standard normal
random vector $\Delta\epsilon\sim\mathcal{N}(0,I_d)$ and we can rely
on the considerations above. The first column of the matrix $A^*$ is
then:
\begin{equation}
\mathbf{A_{\cdot 1}}^* =   \arg\max_{\mathbf{A_{\cdot 1}}\in
\mathbf{R^d}}\left(\frac{\partial
    g^A(\hat{\epsilon})}
    {\partial\epsilon_1}\right)^2
\end{equation}
Since we have already maximized the variance contribution for
$\left(\frac{\partial
g^A(\hat{\epsilon})}{\partial\epsilon_1}\right)^2$, we might
consider the expansion of $g$ about $d-1$ different points in order
to improve the method using adequate columns. More precisely Imai
and Tan \cite{IT2006} propose to maximize:
\begin{equation}\label{4.3.13}
\mathbf{A_{\cdot k}}^* =   \arg\max_{\mathbf{A_{\cdot k}}\in
\mathbf{R^d}}\left(\frac{\partial
    g^A(\hat{\epsilon}_k)}
    {\partial\epsilon_k}\right)^2
\end{equation}
subject to $\|\mathbf{A_{\cdot k}}^*\|=1$ and $\mathbf{A_{\cdot
j}}^*\cdot \mathbf{A_ {\cdot k}}^*=0, j=1,\dots,k-1, k\le d$.

Although equation (\ref{4.3.8}) provides an easy solution at each
step, the correct procedure requires that the column vector
$\mathbf{A_{\cdot k}}^*$ is orthogonal to all the previous (and
future) columns. Imai and Tan \cite{IT2006} propose to choose
$\hat{\epsilon}=\hat{\epsilon}_1=\mathbb{E}[\epsilon]=\mathbf{0}$,
$\hat{\epsilon}_2=(1,0,\dots,0),\dots\hat{\epsilon}_{k}=(1,1,1,\dots,0,\dots,0)$,
where the $k$-th point has $k-1$ leading ones. We refer to Sabino
\cite{Sabino08b,Sabino09} for the details of a fast and convenient
implementation of this algorithm.

\subsection{Enhancing the Localization
Technique}\label{subsec:localization}
The aim of the localization technique introduced in Fournié et al.
\cite{FLLL1999} is to reduce the variance of the MC estimator for
the sensitivities by localizing the integration by part formula
around the singularity. In the following, for simplicity, we
illustrate the localization technique in the case of vanilla call
options.

Fournié et al. \cite{FLLL1999} found that a (possible) expression
for the delta of a call option is:
\begin{equation}
\Delta = \frac{\partial}{\partial
x}\mathbb{E}\left[e^{rT}\left(S(T)-K\right)^+\right]=
\mathbb{E}\left[e^{rT}\left(S(T)-K\right)^+\frac{W(T)}{xT\sigma}\right].
\end{equation}
When the one-dimensional Brownian motion $W(T)$ is large, the term
$\left(S(T)-K\right)^+W(T)$  becomes even larger and has a high
variance. The idea is to introduce a localization function around
the singularity at $K$.

For $\delta>0$, set
\begin{equation}
H_{\delta}(y)=\left\{
\begin{array}{ccc}
  0, & for &y\le K-\delta, \\
  \frac{y-K+\delta}{2\delta} & for & y\in[K-\delta,K+\delta], \\
  1 & for & y\ge K+\delta,
\end{array}\right.
\end{equation}
and $G(z)=\int_{-\infty}^0H_{\delta}(y)dy$, then consider
$F_{\delta}(z)=(z-K)^+-G_{\delta}(z)$. Consequently, we have:
\begin{equation}\label{eq:loc}
  \Delta =e^{rT}\mathbb{E}\left[H_{\delta}(S(T))\frac{\partial S(T)}{\partial x}\right]+
  e^{rT}\mathbb{E}\left[F_{\delta}(S(T))\frac{W(T)}{xT\sigma}\right]
\end{equation}
$F_{\delta}$ vanishes for $z\le K-\delta$ and $z\ge K+\delta$, thus
$F_{\delta}(S(T))W(T)$ vanishes when $W(T)$ is large.

The same analysis, with similar results, is valid for the call-style
Asian options and the exotic option analyzed in Section
\ref{sec:Greeks}. Indeed, it suffices to replace $S(T)$ with the
average $\sum_{i,j}w_{ij}S_{i}(t_j)$ in the equations above and
consider an \emph{if, else} statement to select the localization
function when the strike price is stochastic or the option is
exotic. In addition, in the above options formulas, the role of the
``weight'' term $\frac{W(T)}{xT\sigma}$ is played by the Skorohod
integral. We remark that the formulas that we derived to compute the
$k$-th component of the delta display weights that depend only on
the Skorohod intergral w.r.t. the $k$-th component of the
multidimensional Brownian motion permitting to better control the
variance. If we would have chosen to control all the components of
the Skorohod integral, taking all non-zero components of the simple
vector process $\mathbf{u}$, we would have needed to tune different
$M$ Brownian motions making the localization technique less
efficient and computationally more expansive.

The choice of the parameter $\delta$ is of fundamental importance
for the result of the localization technique because it influences
the variance of the MC estimator. In the following we describe how
to employ an \emph{on the fly} efficient value based on adaptive MC
simulations. For ease of notation, we consider once more a vanilla
call option payoff bearing in mind that the same applies to the
payoffs under study. In such cases we need to make the substitution
illustrated above. A good candidate for $\delta$ would be the one
that minimizes the variance of the second term in equation
(\ref{eq:loc}).
\begin{eqnarray}
\delta^* =  \arg\min_{\delta>0}
Var\left[\frac{F_{\delta}(S(T))W(T)}{x\sigma T}\right]
\end{eqnarray}
and deriving w.r.t. $\delta$:
\begin{equation}
Var\left[-\frac{H_{\delta}(S(T))W(T)}{x\sigma T}\right]=
Var\left[-\frac{W(T)}{ x\sigma
T}\frac{\left(S(T)-K\right)-\delta}{2\delta}\right]=0.
\end{equation}
At this point we find $\delta$ such that:
\begin{equation}
\frac{W(T)}{ x\sigma T}\frac{\left(S(T)-K\right)-\delta}{2\delta} =
0, \quad \mathbb{P}-a.s.
\end{equation}
then
\begin{equation}
\delta = \frac{\frac{\left(S(T)-K\right)W(T)}{x\sigma T}}
{\frac{W(T)}{ x\sigma T}}.
\end{equation}
In order to have an operative parameter we then consider the
following approximation:
\begin{equation}
\delta = \frac{
 Var\left[\frac{\left(S(T)-K\right)W(T)}{x\sigma T}\right]}
 {Var\left[\frac{W(T)}{ x\sigma T}\right]}.
\end{equation}
As already mentioned, the considerations here above are still valid
for the computation of the greeks of the options we are considering.
As already illustrated, it suffices to replace
$\frac{W(T)}{xT\sigma}$ with the Skorohod integral and that is the
reason why we have always shown the term this term explicitly in the
calculations above. The same substitutions must be made to calculate
each $\delta$ for the each component of the Delta of the call type
Asian basket and exotic options since these results hold true in the
multidimensional setting as well.

In the spirit of adaptive MC techniques (see for instance Jourdain
\cite{Jourd09}), the variance above can be easily estimated by a MC
simulation and then, by fixing the same random draws, one runs a
second MC simulation in order to estimate the greeks.

In the case of one dimensional digital options the computation is
slightly different. Kohatsu-Higa and Patterson \cite{KP2002} claim
that a good candidate for $\delta$ is:
\begin{equation}
\delta = \left(\frac{\int_0^{\infty}\phi'(z)^2dz}
{\int_0^{\infty}\phi(z)^2dz
\mathbb{E}\left[\delta^{\text{Sk}}(\mathbf{u})^2\right]}
 \right)^{1/2}.
\end{equation}
Knowing that
$\mathbb{E}\left[\delta^{\text{Sk}}(\mathbf{u})\right]=0$, under the
assumption that $\phi(z)=e^{-|z|}$, we have
\begin{equation}
\delta =
\left(Var\left[\delta^{\text{Sk}}(\mathbf{u})\right]\right)^{-1/2}.
\end{equation}
The above parameter can be easily estimated by an adaptive MC
simulation in the multidimensional setting as illustrated for
call-type options.

We note that in our formulas the computation of the $k$-th delta
depends only on the $k$-th component of the Skorohod integral making
the localization technique easier to apply and the parameter
$\delta$ easy to calculate. Once more, we remark the fact that these
variance and computational reduction considerations are not possible
without using the Malliavin Calculus approach.
\section{Numerical Investigations}\label{sec:NumInv}
In this section we discuss the results of the (R)QMC estimation
based of the proposed approaches. We consider $M=5$ and $M=10$
underlying securities and an equally-spaced time grid with $N=64$
time points. Hence, the effective dimension of the (R)QMC simulation
is either $320$ or $640$. We estimate the multidimensional Deltas
(with respect to each underlying asset) of each contract discussed
before. The parameters chosen for the simulation are listed in Table
\ref{parameters}.
 \begin{table}\centering
 \caption{Inputs Parameters}\label{parameters}
 \begin{tabular}{cc}
     \begin{tabular}{lll}
     \\
         \hline
         $S_{i}(0)$&=&$100, \quad\forall i=1\dots,M$\\
         $r$ &=& $5$\%\\
         $T$&=&$1$\\
         $\sigma_i$ &=&$10\%+\frac{i-1}{9}40\% \quad i=1\dots,M$\\
         $\rho_{il}$&$= 50\% \quad i,l=1\dots,M$\\
         \hline
     \end{tabular}
 \end{tabular}
 \end{table}

We adopt RQMC simulations, based on the enhanced version illustrated
in Section~\ref{subsec:QMC}, and consists of $32$ replications each
of $2048$ random points. These random draws are obtained from a
Matou\^{s}ek affine plus random digital shift scrambled version (see
Matou\^{s}ek~\cite{Ma1998}) of the Sobol sequence satisfying Sobol's
property A (see Sobol~\cite{Sobol76}). We also avoid generating the
$320$ or $640$-dimensional Sobol' sequence by using a Latin
Supercube Sampling (LSS) method (Owen \cite{ow1998B}). Briefly, this
sampling mechanism is a scheme for creating a high-dimensional
sequence from sets of lower-dimensional sequences. For instance, a
$640$-dimensional low discrepancy sequence can be concatenated from
$13$ sets of a $50$-dimensional low discrepancy sequence by
appropriately randomizing the run order of the points (the last
concatenation neglects the last $10$ dimensions). For theoretical
justification of the LSS method, see Owen \cite{ow1998B}.

The computation is implemented in MATLAB on a laptop with an Intel
Pentium M, processor 1.60 GHz and 1 GB of RAM. We compute all the
optimal columns for the LT technique in Section~\ref{subsec:QMC}.
Such an LT construction is optimal if the integrand function is the
payoff of the option and hence is optimal for price estimation. In
contrast, our goal is the computation of the Deltas and this would
not seem to be the optimal choice. However, if we would have applied
the LT for the integrand function given by the Malliavin approach we
would have got as many LT-decomposition matrices as the number of
assets (one for each delta). This setting would remarkably increase
the CPU time making the estimation less convenient. The numerical
experiments below justify our assumption.

\subsection{Call with fix and floating Strike}
As a first experiment, we compute the Deltas of an Asian basket
option with fixed and floating strike. We compare the estimated
values of the Deltas and the accuracies obtained with different
approaches: finite differences, localization with different
parameters and finally localization coupled with adaptive
parameters. The choice of the parameters for the localization and
finite difference techniques is of fundamental importance because it
influences the variance of the estimator (see for instance L'Ecuyer
\cite{LE1995}). The numerical derivative is often calculated
assuming $\delta = 1\%$ (in our case $1\%$ of the initial price of
the underlying securities); this may not be  the optimal choice. In
addition, in the multidimensional computation (gradient estimation)
one should consider different $\delta$. Our approach based on
adaptive techniques overtakes this problem by calculating the
parameters \emph{on the fly}. These parameters are optimal meaning
that they provide the minimal variance of the estimator (in the
sense described in Section \ref{subsec:localization}).
Table~\ref{tab:fixCall} and Table~\ref{tab:floatCall} show the
results with different approaches obtained for an at-the-money Asian
call with fixed and floating strike and $M=10$ underlying assets.
All the estimated values are in statistical accordance but display
different accuracies. The finite difference errors are higher than
those obtained with localization (with the exemption of
$\delta=5\%$). In particular, when the strike is floating, this
technique returns a completed biased Delta associated with the
highest volatility. Finally, finite difference estimations require a
computational effort that  is $2.43$ times higher that those
obtained with localization. The adaptive localization and standard
localization perform equally well with the former having slightly
better precision and the advantage of selecting better localization
parameters for each component.

\begin{table}\centering
    \begin{tabular}{||c|c|c|c|c|c|c|c|c|c||}
        \hline
        \multicolumn{2}{||c|}{Adaptive}&\multicolumn{6}{|c|}{Localization}&\multicolumn{2}{|c||}{Fin.
        Diff.}\\
        \hline
        \multicolumn{2}{||c|}{}&\multicolumn{2}{|c|}{$\delta=1\%$}&\multicolumn{2}{|c|}{$\delta=5\%$}&\multicolumn{2}{|c|}{$\delta=10\%$}&\multicolumn{2}{|c||}{$\delta=1\%$}\\
        \hline
        $\Delta$&$\pm err$&$\Delta$&$\pm err$&$\Delta$&$\pm err$&$\Delta$&$\pm err$&$\Delta$&$\pm err$\\
        \hline
        5.43&0.18&5.43&0.28&5.4&2.9&5.43&0.19&5.43&0.31\\
        5.50&0.23&5.50&0.30&5.5&2.7&5.50&0.26&5.51&0.49\\
        5.58&0.29&5.57&0.30&5.6&2.9&5.58&0.34&5.60&0.52\\
        5.66&0.30&5.65&0.33&5.6&2.9&5.66&0.39&5.69&0.98\\
        5.74&0.39&5.73&0.41&5.7&3.0&5.74&0.35&5.79&0.81\\
        5.82&0.43&5.81&0.44&5.8&3.1&5.83&0.50&5.88&0.87\\
        5.90&0.45&5.89&0.40&5.9&2.9&5.91&0.52&5.99&1.27\\
        5.98&0.35&5.97&0.41&6.0&3.0&6.00&0.51&6.10&0.87\\
        6.07&0.47&6.05&0.40&6.0&3.2&6.09&0.58&6.20&1.40\\
        6.16&0.50&6.13&0.51&6.1&3.0&6.17&0.64&6.29&1.12\\
        \hline
    \end{tabular}
    \caption{Call Option with Fixed Strike, $M=10$: At-the-Money Deltas and Errors  ($\times 100$).}\label{tab:fixCall}
\end{table}
\begin{table}\centering
    \begin{tabular}{||c|c|c|c|c|c|c|c|c|c||}
        \hline
        \multicolumn{2}{||c|}{Adaptive}&\multicolumn{6}{|c|}{Localization}&\multicolumn{2}{|c||}{Fin.
        Diff.}\\
        \hline
        \multicolumn{2}{||c|}{}&\multicolumn{2}{|c|}{$\delta=1\%$}&\multicolumn{2}{|c|}{$\delta=5\%$}&\multicolumn{2}{|c|}{$\delta=10\%$}&\multicolumn{2}{|c||}{$\delta=1\%$}\\
        \hline
        $\Delta$&$\pm err$&$\Delta$&$\pm err$&$\Delta$&$\pm err$&$\Delta$&$\pm err$&$\Delta$&$\pm err$\\
        \hline
        $0.04$&$0.19$&$0.04$&$0.19$&$0.04$&$0.17$&$0.04$&$0.13$&$0.04$&$0.11$\\
        $0.12$&$0.20$&$0.12$&$0.23$&$0.11$&$0.30$&$0.12$&$0.19$&$0.13$&$0.19$\\
        $0.19$&$0.32$&$0.19$&$0.35$&$0.19$&$0.32$&$0.20$&$0.26$&$0.22$&$0.33$\\
        $0.27$&$0.36$&$0.27$&$0.38$&$0.27$&$0.36$&$0.28$&$0.32$&$0.31$&$0.44$\\
        $0.34$&$0.29$&$0.35$&$0.37$&$0.34$&$0.45$&$0.36$&$0.43$&$0.41$&$0.45$\\
        $0.42$&$0.35$&$0.43$&$0.35$&$0.42$&$0.47$&$0.45$&$0.44$&$0.51$&$0.63$\\
        $0.50$&$0.39$&$0.50$&$0.47$&$0.50$&$0.47$&$0.53$&$0.51$&$0.61$&$0.77$\\
        $0.58$&$0.50$&$0.59$&$0.55$&$0.58$&$0.61$&$0.62$&$0.57$&$0.71$&$0.98$\\
        $0.67$&$0.45$&$0.67$&$0.55$&$0.67$&$0.64$&$0.71$&$0.51$&$0.81$&$1.65$\\
        $0.74$&$0.64$&$0.75$&$0.64$&$0.75$&$0.58$&$1.72$&$233.73$&$7\times 10^4$&$2\times 20^7$\\
        \hline
    \end{tabular}
    \caption{Call Option with Floating Strike, $M=10$: Deltas and Errors ($\times 100$).}\label{tab:floatCall}
\end{table}

In order to have a complete picture of the sensitivity of the
discussed techniques, we repeat the experiment considering only
$M=4$ assets and several strike prices. This further analysis cannot
be performed for Asian option with floating strike.
Figure~\ref{fig:deltaFixCall} and \ref{fig:errorFixCall} show the
estimated Deltas and errors, respectively. Since for at-the-money
options the finite difference approach provided lower accuracy, we
avoided to report its results. In term of precision, in this setting
as well, the standard localization with $\delta = 1\%$ and the
adaptive localization return the most accurate results. In
particular, these two approaches perform equally well with the
former one having a more constant trend across all the moneyness.
\begin{figure}
    \centering
    \includegraphics[width=1.1\textwidth]{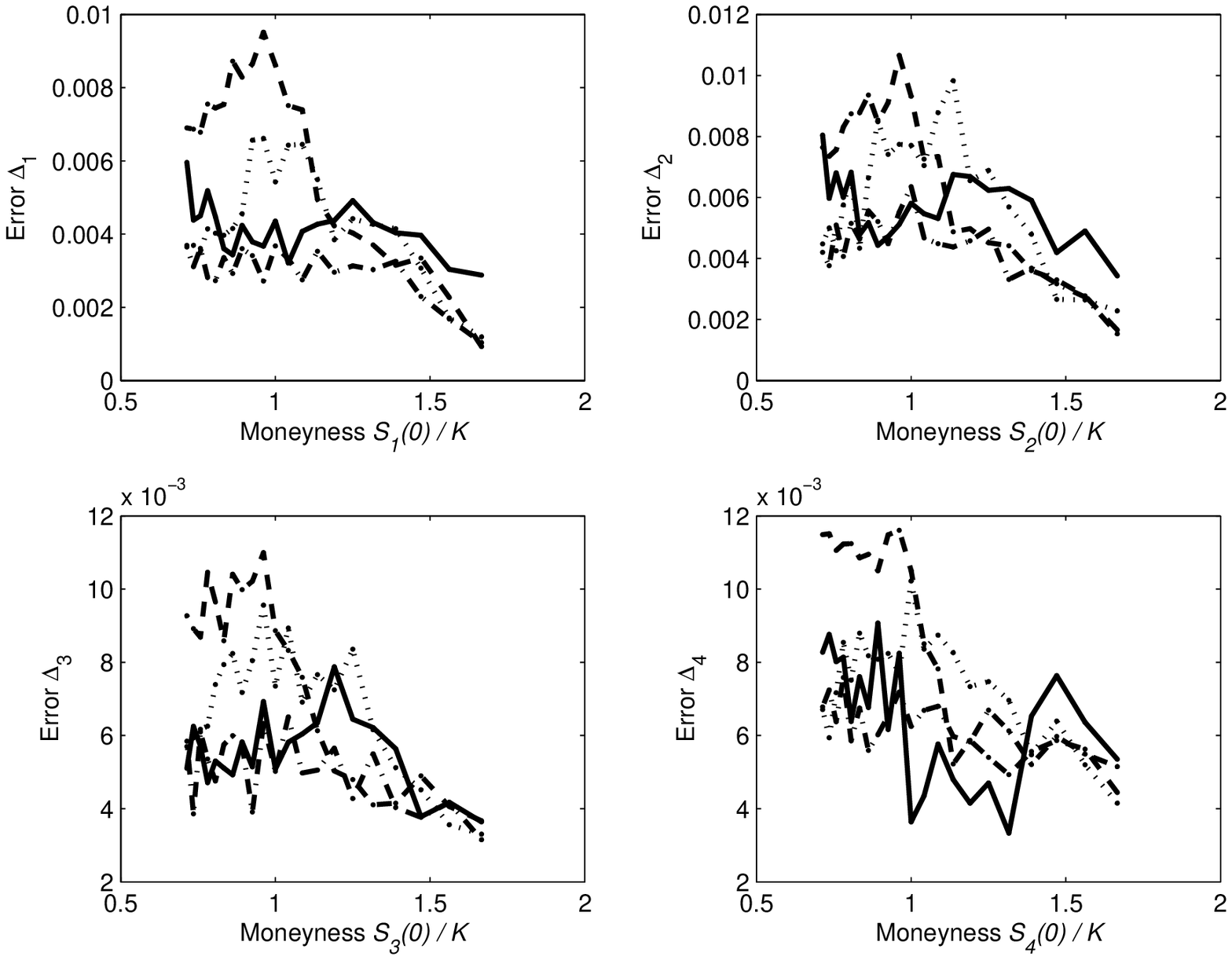}%
    \caption{Call Option with Fixed Strike, $M=4$: Estimation Errors. \newline Adaptive: Solid Line, Loc.
    $\delta=0.01$: Dashed line, Loc. $\delta=0.1$: Dotted line, Loc. $\delta=0.05$: Dash-dotted line.}
    \label{fig:errorFixCall}
\end{figure}

\begin{figure}
    \centering
    \includegraphics[width=1.1\textwidth]{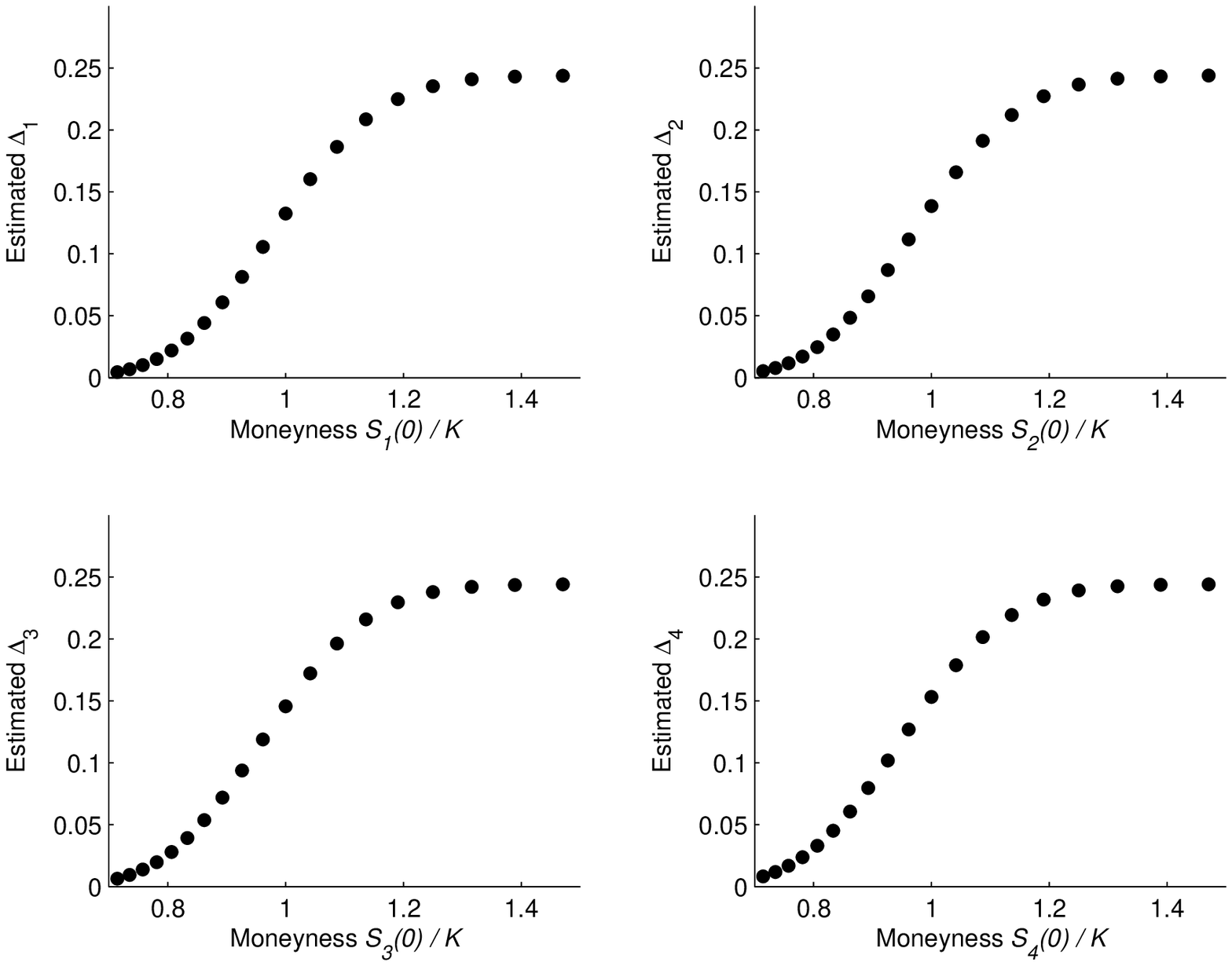}%
    \caption{Call Option with Fixed Strike, $M=4$: Estimated Deltas with
    the Adaptive Localization.}\label{fig:deltaFixCall}
\end{figure}

\subsection{Digital Call}
The aim of this subsection is to describe the results of our
numerical investigation assuming Asian digital options. The
following discussion and description have a double purpose. Since
the payoff of digital option can be seen as the derivative (in the
sense of distribution) of the payoff of a call option, the
methodology and the localization parameters described in
Section~\ref{subsec:BS} can be be rearranged and used to compute the
Gamma (and cross sensitivities in the multidimensional setting) of a
call option (naturally with some changes). In addition, the Delta of
a digital option is a more demanding task due to the irregular
payoff that is pathologically not differentiable.

We repeat the organization of our discussion as done for the Asian
call options and consider only a fixed strike price. Table
~\ref{tab:dig} shows the estimated multidimensional Deltas and their
errors for an at-the-money digital option on $M=10$ underlying
securities. The best accuracy with the standard localization
technique is not achieved anymore with $\delta = 1\%$, that means
that in some situations it is not the optimal choice. In contrast,
the adaptive localization is the best performing technique in terms
of precision. It returns better localization parameters that provide
an unbiased estimator with lower variance.

\begin{table}\centering
    \begin{tabular}{||c|c|c|c|c|c|c|c|c|c||}
        \hline
        \multicolumn{2}{||c|}{Adaptive}&\multicolumn{6}{|c|}{Localization}&\multicolumn{2}{|c||}{Fin.
        Diff.}\\
        \hline
        \multicolumn{2}{||c|}{}&\multicolumn{2}{|c|}{$\delta=1\%$}&\multicolumn{2}{|c|}{$\delta=5\%$}&\multicolumn{2}{|c|}{$\delta=10\%$}&\multicolumn{2}{|c||}{$\delta=1\%$}\\
        \hline
        $\Delta$&$\pm err$&$\Delta$&$\pm err$&$\Delta$&$\pm err$&$\Delta$&$\pm err$&$\Delta$&$\pm err$\\
        \hline
        $0.30$&$0.15$&$0.30$&$0.75$&$0.31$&$3.8$&$0.30$&$0.18$&$0.30$&$0.19$\\
        $0.29$&$0.23$&$0.29$&$0.86$&$0.31$&$3.5$&$0.29$&$0.31$&$0.29$&$0.27$\\
        $0.29$&$0.29$&$0.30$&$0.74$&$0.31$&$3.6$&$0.29$&$0.31$&$0.28$&$0.34$\\
        $0.29$&$0.45$&$0.29$&$0.80$&$0.30$&$3.2$&$0.28$&$0.38$&$0.28$&$0.45$\\
        $0.29$&$0.48$&$0.29$&$0.78$&$0.31$&$3.6$&$0.28$&$0.41$&$0.26$&$0.49$\\
        $0.29$&$0.56$&$0.29$&$0.88$&$0.31$&$3.4$&$0.27$&$0.49$&$0.25$&$0.49$\\
        $0.28$&$0.56$&$0.28$&$0.82$&$0.31$&$3.6$&$0.27$&$0.64$&$0.24$&$0.65$\\
        $0.27$&$0.55$&$0.28$&$0.98$&$0.30$&$3.6$&$0.25$&$0.48$&$0.23$&$0.50$\\
        $0.27$&$0.58$&$0.28$&$0.80$&$0.29$&$3.9$&$0.24$&$0.72$&$0.22$&$0.77$\\
        $0.27$&$0.60$&$0.27$&$0.86$&$0.30$&$3.4$&$0.24$&$0.70$&$0.20$&$0.60$\\
        \hline
    \end{tabular}
    \caption{Digital Option with Fixed Strike, $M=10$: At-the-Money Deltas and Errors.}\label{tab:dig}
\end{table}
As done before, we run a QMC simulation considering only $M=4$
assets and analyze the results by varying the strike price.
Figure~\ref{fig:errorDig} and \ref{fig:deltaDig} show the estimated
Deltas and errors, respectively. Once more the adaptive localization
approach displays the lowest error.

\begin{figure}
    \centering
    \includegraphics[width=1.1\textwidth]{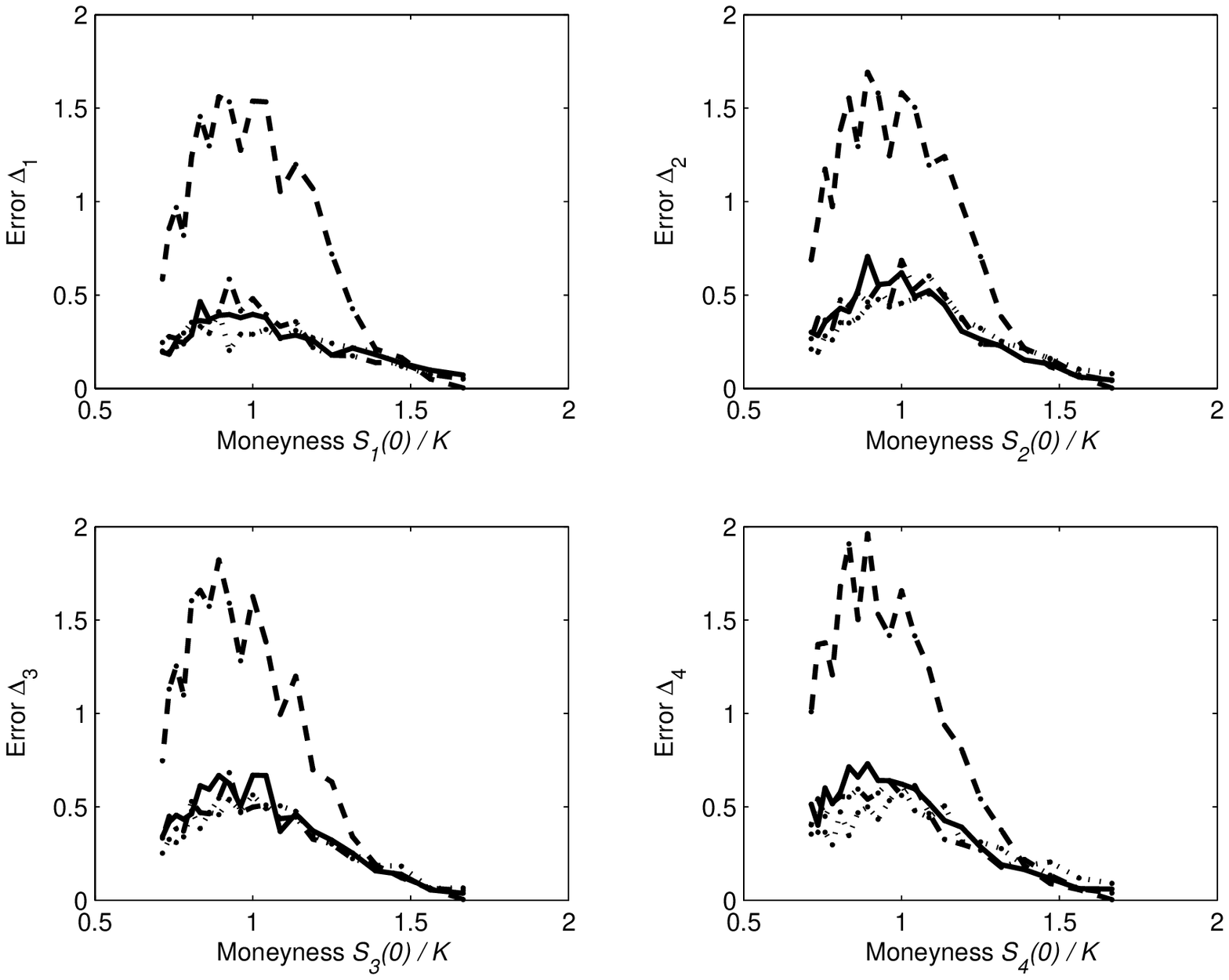}%
    \caption{Digital Option with Fixed Strike, $M=4$: Estimation Errors. \newline Adaptive: Solid Line, Loc.
    $\delta=0.01$: Dashed line, Loc. $\delta=0.1$: Dotted line, Loc. $\delta=0.05$: Dash-dotted line.}
    \label{fig:errorDig}
\end{figure}

\begin{figure}
    \centering
    \includegraphics[width=1.1\textwidth]{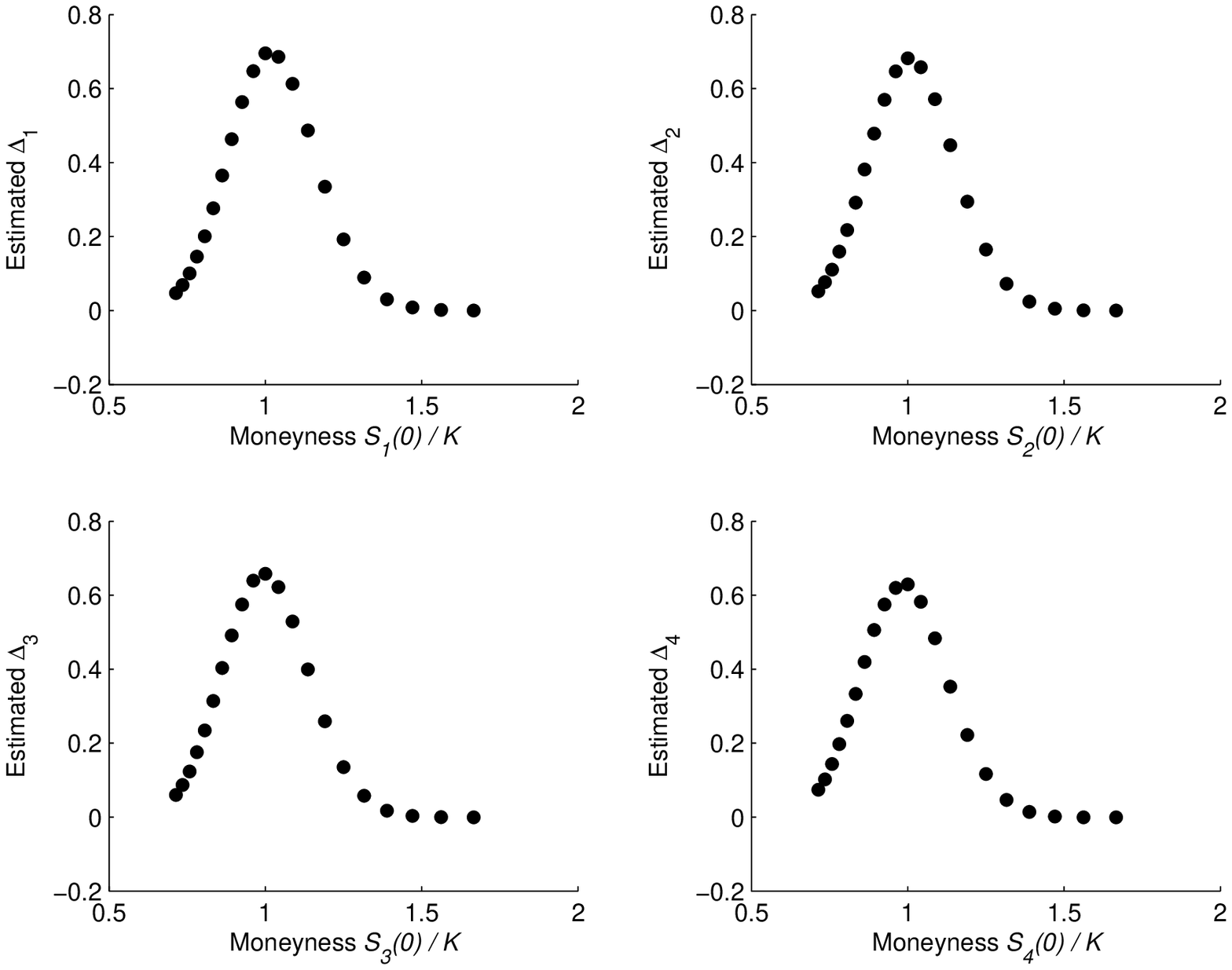}%
    \caption{Digital Option with Fixed Strike, $M=4$: Estimated Deltas with
    the Adaptive Localization.}\label{fig:deltaDig}
\end{figure}

\subsection{Exotic Option}
As a last experiment we perform a QMC numerical simulation in order
to estimate the Deltas of an exotic option. Table~\ref{tab:exot} and
Figures \ref{fig:errorExot} and \ref{fig:deltaExot} present the
results of this experiment. In this last example all the approaches
perform equally well, and the exotic structure of the payoff makes
its estimator unsensitive to the different localization parameter.
The finite difference is also performing well but is less precise if
we take into account the computational burden that is $2.61$ times
higher.

\begin{table}\centering
    \begin{tabular}{||c|c|c|c|c|c|c|c|c|c||}
        \hline
        \multicolumn{2}{||c|}{Adaptive}&\multicolumn{6}{|c|}{Localization}&\multicolumn{2}{|c||}{Fin.
        Diff.}\\
        \hline
        \multicolumn{2}{||c|}{}&\multicolumn{2}{|c|}{$\delta=1\%$}&\multicolumn{2}{|c|}{$\delta=5\%$}&\multicolumn{2}{|c|}{$\delta=10\%$}&\multicolumn{2}{|c||}{$\delta=1\%$}\\
        \hline
        $\Delta$&$\pm err$&$\Delta$&$\pm err$&$\Delta$&$\pm err$&$\Delta$&$\pm err$&$\Delta$&$\pm err$\\
        \hline
        $6.4$&$1.1$&$6.5$&$1.3$&$6.5$&$1.8$&$6.5$&$0.8$&$6.5$&$1.0$\\
        $6.6$&$1.6$&$6.6$&$1.2$&$6.6$&$1.8$&$6.6$&$1.0$&$6.6$&$1.1$\\
        $6.8$&$1.3$&$6.7$&$1.2$&$6.7$&$1.9$&$6.7$&$0.9$&$6.8$&$1.3$\\
        $7.0$&$1.5$&$6.9$&$1.4$&$6.9$&$1.7$&$6.9$&$1.2$&$6.9$&$2.0$\\
        $7.2$&$1.9$&$7.0$&$1.1$&$7.0$&$1.8$&$7.1$&$1.1$&$7.1$&$2.0$\\
        $7.4$&$1.5$&$7.2$&$1.4$&$7.2$&$1.8$&$7.2$&$1.3$&$7.3$&$2.2$\\
        $7.6$&$1.7$&$7.3$&$1.4$&$7.3$&$1.9$&$7.4$&$1.3$&$7.5$&$2.1$\\
        $7.8$&$1.9$&$7.5$&$1.5$&$7.5$&$1.6$&$7.5$&$1.7$&$7.7$&$2.1$\\
        $8.0$&$2.1$&$7.6$&$1.5$&$7.6$&$1.7$&$7.7$&$1.4$&$7.9$&$2.4$\\
        $8.2$&$1.7$&$7.8$&$1.7$&$7.8$&$1.8$&$7.9$&$1.5$&$8.1$&$1.8$\\
        \hline
    \end{tabular}
    \caption{Exotic, $M=10$: At-the-Money Deltas and Errors ($\times
    100$).}\label{tab:exot}
\end{table}

\begin{figure}
    \centering
    \includegraphics[width=1.1\textwidth]{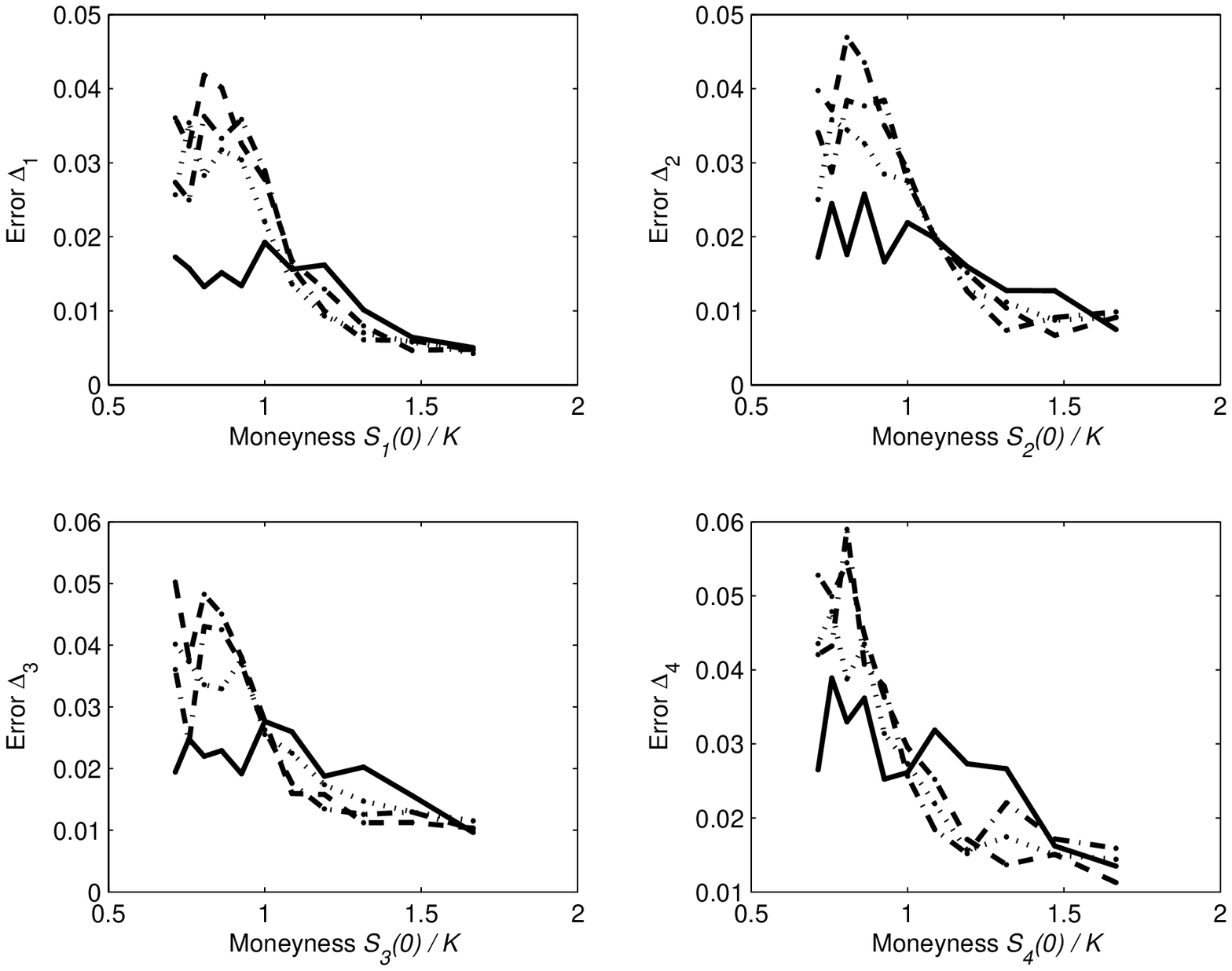}%
    \caption{Exotic Option, $M=4$: Estimation Errors. \newline Adaptive: Solid Line, Loc.
    $\delta=0.01$: Dashed line, Loc. $\delta=0.1$: Dotted line, Loc. $\delta=0.05$: Dash-dotted line.}
    \label{fig:errorExot}
\end{figure}

\begin{figure}
    \centering
    \includegraphics[width=1.1\textwidth]{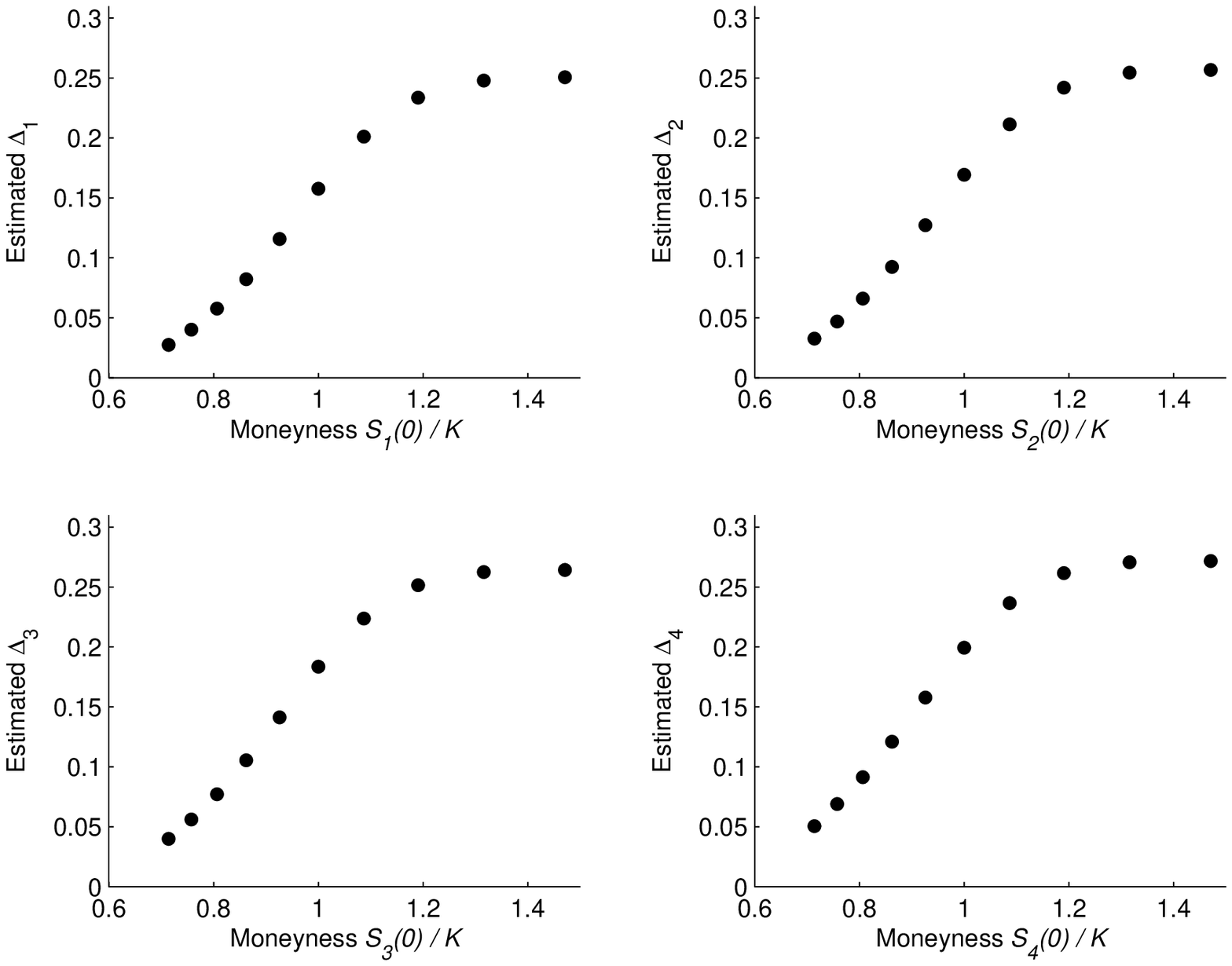}%
    \caption{Exotic Option, $M=4$: Estimated Deltas with
    the Adaptive Localization.}\label{fig:deltaExot}
\end{figure}

\section{Concluding Remarks}\label{sec:conclusions}
In this paper we have investigated the use of Malliavin calculus in
order to calculate the Greeks of multiasset complex path-dependent
options by QMC simulation. As a first result we have derived the
multidimensional version of the formulas obtained by Montero and
Kohatsu-Higa \cite{MKH2004} in the single asset case.The
multidimensional setting shows the advantage of the Malliavin
Calculus approach over alternative techniques that have been
previously proposed. These different techniques are hard to
implement and in particular, are computationally time consuming when
considering multiasset derivative securities. In addition, their
estimators potentially display a high variance (see for instance
Chen and Glasserman \cite{CG07}). In contrast, the use of the
generalized integral by part formula of Malliavin Calculus gives
enough flexibility in order to find unbiased estimators with low
variance. In the multidimensional context, we have found convenient
formulas that are easy and flexible to employ and permit to improve
the localization technique. Finally, we have performed a detailed
analysis on how the localization parameters can influence the
precision of the estimators. Moreover, we have proposed an alternate
approach, based on adaptive (Q)MC techniques that returns convenient
parameters that can be obtained \textit{on the flight} in the
simulation. This approach provides a better  precision with the same
computational burden. However further studies would be necessary to
enhance its accuracy assuming different dynamics and payoff
functions.

The proposed procedures, coupled with the enhanced version of
Quasi-Monte Carlo simulations as illustrated in Sabino
\cite{Sabino08b}, are discussed based on the numerical estimation of
the Deltas of call, digital Asian-style and Exotic basket options
with a fixed and a floating strike price in a multidimensional
Black-Scholes market.

\newpage
\bibliographystyle{QMCMalliavin}
\bibliography{tesi}

\end{document}